\documentclass[11pt]{article}
\usepackage{amsmath}
\usepackage{amssymb}
\usepackage{amsthm}
\usepackage{graphicx}
\usepackage{booktabs}
\usepackage{caption}
\usepackage{subcaption}

\newcommand{\Var}{\mathrm{Var}}

\newcommand{\D}{\mathrm{d}}
\newcommand{\E}{\mathbb E}
\newcommand{\Q}{\mathbb Q}
\newcommand{\F}{\mathcal F}

\theoremstyle{theorem}
\newtheorem{theorem}{Theorem}[section]
\newtheorem{lemma}[theorem]{Lemma}

\theoremstyle{definition}
\newtheorem{assumption}[theorem]{Assumption}
\newtheorem{definition}[theorem]{Definition}
\newtheorem{remark}[theorem]{Remark}

\begin{document}

\title{Risk-neutral option pricing under GARCH intensity model}
\author{Kyungsub Lee\footnote{Department of Statistics, Yeungnam University, Gyeongsan, Gyeongbuk 38541, Korea}}
\date{}
\maketitle

\begin{abstract}
The risk-neutral option pricing method under GARCH intensity model is examined.
The GARCH intensity model incorporates the characteristics of financial return series such as volatility clustering, leverage effect and conditional asymmetry.
The GARCH intensity option pricing model has flexibility in changing the volatility according to the probability measure change.

\end{abstract}

\section{Introduction}

This paper develop the risk-neutral option pricing framework under the GARCH intensity model.
The financial asset returns series have interesting characteristics.
Large volatility tends to follow large volatility and small volatility tends to follow small volatility in the series of financial returns and this is called volatility clustering.
The leverage effect indicates that today's volatility has a negative correlation with past returns.
Conditional asymmetry is an asymmetric correlation between current and past volatility, depending on whether the current and past returns are positive or negative.

These characteristics are well captured by various GARCH models.
The volatility clustering is well described by original GARCH model \cite{Bollerslev}. 
\cite{Pagan&Schwert}, \cite{Nelson},  \cite{GJR}, and \cite{Zakoian} incorporate the leverage effects and \cite{Babsiri} capture the conditional asymmetry.
The GARCH intensity model introduced by \cite{ChoeLee} also describes the volatility clustering, leverage effect and conditional asymmetry in financial asset price dynamics and based on Poisson type intensity processes.

The risk-neutral option pricing is a method to determine a no-arbitrage prices of financial options with underlying assets.
The option pricing theory of \cite{BlackScholes} and \cite{merton1973theory} is based on the risk-neutral pricing framework.
The relation between risk-neutral pricing and the no-arbitrage principle was studied by \cite{HarrisonKreps} and \cite{HarrisonPliska}.
\cite{Duan1995} explained the risk-neutral option pricing under the GARCH model.
This paper extends the risk-neutral pricing idea to the GARCH intensity model.

The remainder of the paper is organized as follows:
Section \ref{Sect:model} reviews the GARCH intensity model.
Section \ref{Sect:option} examines the mathematical analysis on the risk-neutral option pricing under the GARCH intensity model.
Section \ref{Sect:Generalization} extends the option pricing theory to the generalized GARCH intensity model.
Section \ref{Sect:concl} concludes the paper.

\section{GARCH intensity model}\label{Sect:model}
First, we review the GARCH intensity model introduced by~\cite{ChoeLee}.
In the model, the asset price process movement is described by two Poisson-type processes with time-varying intensity processes.
A probability space $(\Omega, \mathcal F = \mathcal F(T), \mathbb P)$ with a filtration $\mathcal F(t)$, $0\leq t\leq T$, is given.

\begin{assumption}[\cite{ChoeLee}]\label{Assumption}
We are given $\mathcal F(t)$-adapted r.c.l.l. processes $N_{+}(t)$, $N_{-}(t)$
and positive $\mathcal F(t)$-adapted r.c.l.l. processes $\lambda_{+}(t)$, $\lambda_{-}(t)$
for $0 \leq t \leq T$ satisfying the following conditions:\\
(i) (Discrete observation time) $\Delta t = T/N$ and $t_i = i\Delta t$, $0 \leq i \leq N$.\\
(ii) (Conditional distribution)
$\left( N_{\pm}(t) - N_{\pm}(t_{i-1}) \right) | \mathcal F(t_{i-1})$ has Poisson distribution
with intensity $\lambda_{\pm}(t_{i-1})( t- t_{i-1})$,
 $t_{i-1} \leq t \leq t_{i}$.
Hence
$$\mathbb P( N_{\pm}(t_i) -  N_{\pm}(t_{i-1})= k | \mathcal F(t_{i-1}))
=  \frac{(\lambda_{\pm}(t_{i-1}) \Delta t)^k}{ k!}\exp{(-\lambda_{\pm}(t_{i-1}) \Delta t )}.$$
(iii) (Conditional independence)
$N_+(t) - N_+(t_{i-1})$ and $N_-(t) - N_-(t_{i-1})$ are conditionally independent given
$\mathcal F(t_{i-1})$, $t_{i-1} \leq t \leq t_{i}$.\\
(iv) (Step process)
$\lambda_{+}(t) = \lambda_{+}(t_{i-1})$ and $\lambda_{-}(t) = \lambda_{-}(t_{i-1})$, $t_{i-1} \leq t < t_{i}$.\\
(v) (Predictability) $\lambda_+(t_i)$  depends on $N_{\pm}(t_{i-k+1})$ and $\lambda_{\pm}(t_{i-k})$, $1\leq k \leq i-1$,
and similarly for $\lambda_-(t_i)$.\\
(vi) (Asset price) With a constant $\delta>0$, the price process is
$$S(t)= S(0)\exp\left( \delta ( N_{+}(t) - N_{-}(t))\right).$$
\end{assumption}

With a price jump at time $t$, $S(t)=e^{\delta}S(t-)$ or $S(t)=e^{-\delta}S(t-)$ depending on the direction of the jump.
The asset price can also be represented by a stochastic differential equation given by
$$\D S(t)= (e^\delta -1)S(t-)\D N_{+} (t) + (e^{-\delta}-1)S(t-)\D N_{-} (t).$$
Let
$$X(t_i)= \log  \frac{S(t_{i})}{S(t_{i-1})}$$
be the log-return over the period $[t_{i-1},t_i]$.
Then the integer-valued random variable $M_i$ defined by
$$M_i=\frac{X(t_i)}{\delta}=N_{+}(t_i)-N_{-}(t_i)-(N_{+}(t_{i-1})-N_{-}(t_{i-1}))$$ has the conditional Skellam distribution on $\mathcal F(t_{i-1})$
\begin{eqnarray*}
& & f(m | \lambda_{+}(t_{i-1}),\lambda_{-}(t_{i-1})) \\
&=& \exp \{ -\lambda_{+}(t_{i-1}) - \lambda_{-}(t_{i-1}) \}
\left( \frac{\lambda_{+}(t_{i-1})}{\lambda_{-}(t_{i-1})} \right)^{m/2} I_{|m|}( 2\sqrt{ \lambda_{+}(t_{i-1})\lambda_{-}(t_{i-1})} )
\end{eqnarray*}
where $I_{a}$ is the modified Bessel function of the first kind defined by
$$ I_{a}(x) = \sum_{k=0}^{\infty} \frac{1}{ k! \,\Gamma(k+a+1)} \left( \frac{x}{2} \right)^{2k+a}.$$
Since the closed form of conditional probability density is exist, the maximum likelihood estimation can be easily employed.

\begin{definition}[Decomposition of Log-Return]\label{Def:components}
Define $\mu(t_i)$, $\gamma(t_i)$, $\varepsilon(t_i)$ by
\begin{eqnarray*}
\mu(t_i) &=& \{ (e^\delta - 1) \lambda_{+} (t_{i-1}) + (e^{-\delta} - 1)\lambda_{-} (t_{i-1}) \} \Delta t \\
\gamma(t_i) &=& \{ (e^\delta - 1 -  \delta)  \lambda_+ (t_{i-1}) + (e^{-\delta} - 1 +  \delta) \lambda_- (t_{i-1}) \} \Delta t \\
\varepsilon(t_i) &=& X(t_i) - \mathbb E [X(t_i) | \mathcal F(t_{i-1})].
\end{eqnarray*}
\end{definition}

Recall that under Assumption~\ref{Assumption}, 
\begin{eqnarray*}
\mathbb E[X(t_i)|\mathcal F(t_{i-1})]  &=& \delta(\lambda_{+}(t_{i-1})- \lambda_{-}(t_{i-1})) \Delta t\\
\Var(X(t_i)|\mathcal F(t_{i-1}))&=& \delta^{2}(\lambda_{+}(t_{i-1})+\lambda_{-}(t_{i-1})) \Delta t\\
\mathbb E[ \exp(X(t_i))|\mathcal F(t_{i-1}))&=& \exp(\mu(t_i))
\end{eqnarray*}
and
$$X(t_i) = \mu(t_i) - \gamma(t_i) + \varepsilon(t_i),$$
where $\mu(t_i)$ is a drift term, $\gamma(t_i)$ is an It\^{o} correction factor,
and $\varepsilon(t_i)$ is a $\mathcal F(t_{i})$-measurable shock occurred during time interval $[t_{i-1}, t_i]$.

As in \cite{ChoeLee}, to capture volatility clustering, GARCH\cite{Bollerslev}-type modeling is applied:
$$
\lambda_{\pm}(t_i) = \omega_{\pm} + \alpha_{\pm}  \varepsilon^2(t_{i}) + \beta_{\pm} \lambda_{\pm}(t_{i-1})
$$
for some constants $\omega_{\pm}, \alpha_{\pm},$ and $\beta_{\pm}$.
If 
$$\beta_{+} = \beta_{-}=\beta$$ in the GARCH intensity model, 
then the GARCH-type time varying volatility is obtained.
To show this, let $h(t_i)$ be a one-step-ahead conditional variance of return at $t_i$, then
$$h(t_i) = \Var(X(t_i)|\mathcal F(t_{i-1}))/\Delta t = \delta^2(\lambda_{+}(t_{i-1}) + \lambda_{-}(t_{i-1})),$$
and
\begin{equation}\label{Eq:variance-intensity}
h(t_i)
= \delta^{2}(\omega_{+} + \omega_{-}) + \beta h(t_{i-1}) + {\delta^2}(\alpha_{+} + \alpha_{-} )\varepsilon^2(t_{i-1}).
\end{equation}
This is consistent with conditional variance modeling in GARCH.
In addition, we also consider the GJR\cite{GJR} GARCH-type intensity model:
$$
\lambda_{\pm}(t_i) = \omega_{\pm} + \alpha_{\pm}  \varepsilon^2(t_{i}) + \beta_{\pm} \lambda_{\pm}(t_{i-1})
$$
where
$$I(t_{i})=\left\{
  \begin{array}{ll}
    1, & \hbox{$\varepsilon(t_{i})<0$} \\
    0, & \hbox{$\varepsilon(t_{i})\geq 0$}.
  \end{array}
\right.
$$

\section{Risk-neutral option pricing}\label{Sect:option}

We propose an option pricing method for intensity models
by constructing an equivalent measure under which the discounted stock price process is a martingale.
We assume that the underlying asset pays no dividend and let $r>0$ be the risk-free interest rate.
In the following we choose new intensities for an equivalent martingale measure.

\begin{definition}\label{Def:RadonNikodym}
Take a pair of positive  r.c.l.l. adapted step processes $\widetilde\lambda_{+}$ and $\widetilde\lambda_{-}$
such that
\begin{eqnarray}\label{Eq:risk-free-rate}
(e^\delta -1)\widetilde\lambda_{+}(t) + (e^{-\delta} -1)\widetilde\lambda_{-}(t) = r .
\end{eqnarray}
(We take the right hand side equal to $r$ since the left hand side is regarded as drift under a risk-neutral measure.)\\
(i) Let
$$D(t) = \lambda_{+}(t) + \lambda_{-}(t) - \widetilde\lambda_{+}(t) - \widetilde \lambda_{-}(t)$$
for $0 \leq t \leq T$, and let $U(0) = 0$ and for $t_{i-1} < t \leq t_{i}$ define
$$
 U(t) = (N_+(t) - N_+(t_{i-1})) \log \frac{\widetilde\lambda_{+}(t) }{\lambda_{+}(t)  }
+ (N_-(t) - N_-(t_{i-1})) \log \frac{\widetilde\lambda_{-}(t) }{\lambda_{-}(t) } .$$
(ii)
Let
$Z(0) =1$
and for $t_{i-1} < t \leq t_{i}$ define
$$Z(t) = Z(t_{i-1}) \exp \{ D(t) (t-t_{i-1}) +  U(t) \}.$$
\end{definition}

\begin{theorem}\label{Thm:RadonNikodym}
$\{ Z(t) \}_{ 0 \leq t \leq T } $ is a $\mathbb P$-martingale.
\end{theorem}
\begin{proof}
For $t_{i-1} < t \leq t_i$, define $Z_{+,i}(t)$ and $Z_{-,i}(t)$ by $Z_{\pm}(t_{i-1}) = 1$ and
$$ \D Z_{\pm,i}(t) = Z_{\pm,i}(t-) \frac{\widetilde\lambda_{\pm}(t_{i-1})
- \lambda_{\pm}(t_{i-1})}{\lambda_{\pm}(t_{i-1})}  \D ( N_{\pm}(t) - \lambda_{\pm}(t_{i-1})t).$$
Then
\begin{eqnarray*}
Z_{\pm,i}(t) &=& \exp \bigg\{ (\lambda_{\pm}(t_{i-1}) - \widetilde\lambda_{\pm}(t_{i-1}))(t-t_{i-1}) \\
&+&  (N_{\pm}(t) - N_{\pm}(t_{i-1}) ) \log \frac{\widetilde\lambda_{\pm}(t_{i-1}) }{\lambda_{\pm}(t_{i-1}) } \bigg\}.
\end{eqnarray*}
(For the details of the proof, see \cite{Elliott&Kopp}.)
Since $N_{\pm}(t) - \lambda_{\pm}(t_{i-1})t$ are martingales,
$Z_{\pm ,i}(t)$ are martingales for $t_{i-1} \leq u < t$
and hence
$$\mathbb E[Z_{\pm ,i}(t)| \mathcal F(t_{i-1})]=1.$$
Note that
$$Z_{+,i}(t)Z_{-,i}(t) =  \frac{Z(t)}{Z(t_{i-1})}.$$
Since $N_{+}(t) - N_{+}(t_{i-1})$ and $N_{-}(t) - N_{-}(t_{i-1})$ are conditionally independent
given $\mathcal F(t_{i-1})$, we have
$$
\mathbb E\left[\frac{Z(t)}{Z(t_{i-1})} \bigg| \mathcal F(t_{i-1})\right]
=\mathbb E[Z_{+,i}(t) | \mathcal F(t_{i-1})] \; \mathbb E[Z_{-,i}(t) | \mathcal F(t_{i-1})] =1.
$$
Take $s,t$ such that $t_{j-1} < s \leq t_j$ and $s<t$.
Then
$$ \mathbb E[ Z(t_{j}) | \mathcal F(s)] = Z(s)$$
and
\begin{eqnarray*}
& &\mathbb E[Z(t) | \mathcal F(s)]\\
&=& \mathbb E\left[\frac{Z(t)}{Z(t_{i-1})}\frac{Z(t_{i-1})}{Z(t_{i-2})} \times \dots \times \frac{Z(t_{j+1})}{Z(t_{j})}Z(t_{j}) \bigg| \mathcal F(s) \right] \\
&=& \mathbb E \left[\mathbb E\left[\frac{Z(t)}{Z(t_{i-1})} \frac{Z(t_{i-1})}{Z(t_{i-2})} \times \dots \times \frac{Z(t_{j+1})}{Z(t_{j})}Z(t_{j}) \bigg| \mathcal F(t_{i-1})\right] \bigg| \mathcal F(s) \right] \\
&=& \mathbb E \left[\frac{Z(t_{i-1})}{Z(t_{i-2})} \times \dots \times \frac{Z(t_{j+1})}{Z(t_{j})}Z(t_{j}) \mathbb E\left[\frac{Z(t)}{Z(t_{i-1})} \bigg| \mathcal F(t_{i-1})\right] \bigg| \mathcal F(s) \right] \\
&=& \mathbb E\left[\frac{Z(t_{i-1})}{Z(t_{t-2})} \times \dots \times \frac{Z(t_{j+1})}{Z(t_{j})}Z(t_{j}) \bigg| \mathcal F(s) \right] \\
& & \qquad \vdots\\
&=& \mathbb E[Z(t_j)| \mathcal F(s)]\\
&=& Z(s).
\end{eqnarray*}
\end{proof}

\begin{definition}\label{Def:Q}
Define an equivalent probability measure $\mathbb Q$ by
$$\mathbb Q(A) = \int_{A} Z(T) \D \mathbb P \quad \textrm{for }A \in \mathcal F.$$
\end{definition}

Now we change intensities.
\begin{lemma}\label{Lemma:Change_of_intensities}
The intensities of $N_{+}$ and $N_{-}$ under $\mathbb Q$ are
given by $\widetilde\lambda_{+}$ and $\widetilde\lambda_{-}$, respectively.
\end{lemma}

\begin{proof}
Since $\{N_+(t) - N_+(u)\}| \mathcal F(u) $ has a Poisson distribution for $t_{i-1}\leq u <t\leq t_i$, we have
$$
\mathbb E^{\mathbb P} \left[ \exp\left\{\eta_i ( N_+(t) - N_+(u) )  \right\} \big| \mathcal F(u)\right]
= \exp \{\lambda_+(t_{i-1}) (t-u) (e^{\eta_i} - 1) \}
$$
for an $\mathcal F(t_{i-1})$-measurable random variable $\eta_i$.
We will show that the same relation holds for $\mathbb Q$ and $\widetilde{\lambda}_{+}$.
Define $Z_{\pm,i}(t)$ as in the proof of Theorem~\ref{Thm:RadonNikodym}.
For a constant $\xi$, we have
\begin{eqnarray*}
& & \mathbb E^{\mathbb Q}[ \exp\{ \xi (N_+(t)-N_+(t_{i-1})) \} | \mathcal F(t_{i-1})] \\
&=& \mathbb E^{\mathbb P} \left[ \exp \left\{ \xi (N_+(t)-N_+(t_{i-1})) \right\} \frac{Z(t)}{Z(t_{i-1})} \bigg| \mathcal F(t_{i-1})\right]\\
&=& \mathbb E^{\mathbb P}[ \exp\{\xi (N_+(t)-N_+(t_{i-1}))\} Z_{+,i}(t)Z_{-,i}(t) | \mathcal F(t_{i-1})]\\
&=& \mathbb E^{\mathbb P}[ \exp\{\xi (N_+(t)-N_+(t_{i-1}))\} Z_{+,i}(t) | \mathcal F(t_{i-1})] \;\mathbb E^{\mathbb P}[ Z_{-,i}(t)| \mathcal F(t_{i-1}) ]\\
&=& \mathbb E^{\mathbb P} [ \exp\left\{\xi (N_+(t)-N_+(t_{i-1}))\right\} Z_{+,i}(t) | \mathcal F(t_{i-1})]\\
&=& \mathbb E^{\mathbb P} \bigg[ \exp\{\xi (N_+(t)-N_+(t_{i-1})) + (\lambda_+(t_{i-1}) - \widetilde\lambda_+ (t_{i-1}))(t-t_{i-1})\}  \\
& &\times\quad \exp\left\{ ( N_+(t) - N_+(t_{i-1})) \log \frac{\widetilde\lambda_{+}(t_{i-1}) }{\lambda_{+}(t_{i-1})} \right\} \bigg| \mathcal F(t_{i-1})\bigg] \\
&=& \exp \{ (\lambda_+ (t_{i-1}) - \widetilde\lambda_+ (t_{i-1}))(t-t_{i-1}) \} \\
& &\times\quad \mathbb E^{\mathbb P} \left[ \exp\left\{( N_+(t) - N_+(t_{i-1}))\left( \xi +  \log \frac{\widetilde\lambda_{+}(t_{i-1}) }{\lambda_{+}(t_{i-1})}\right) \right\} \bigg| \mathcal F(t_{i-1})\right]\\
&=& \exp \{ (\lambda_+ (t_{i-1}) - \widetilde\lambda_+ (t_{i-1}))(t-t_{i-1}) \} \\
& &\times \exp \left\{ \lambda_+(t_{i-1}) \left( \frac{\widetilde\lambda_{+}(t_{i-1})} {\lambda_{+}(t_{i-1}) } e^{\xi } - 1\right)(t-t_{i-1})  \right\} \\
&=& \exp \{ \widetilde\lambda_+(t_{i-1})(e^\xi  - 1)(t-t_{i-1}) \}
\end{eqnarray*}
where the last expression is the moment generating function of a Poisson distribution with intensity $\widetilde \lambda_+(t_{i-1})(t-t_{i-1})$.
For $N_-(t) - N_-(t_{i-1})$, the proof is identical.
\end{proof}

As in Definition~\ref{Def:components} we define $\widetilde{\gamma}(t_i)$ and $\widetilde{\varepsilon}(t_i)$ in terms of
$\widetilde\lambda_{+}$, $\widetilde\lambda_{-}$ and $\mathbb Q$.
\begin{eqnarray*}
\widetilde{\gamma}(t_i) &=& \{ (e^\delta - 1 -  \delta)  \widetilde{\lambda}_{+} (t_{i-1}) + (e^{-\delta} - 1 + \delta) \widetilde{\lambda}_{-} (t_{i-1}) \} \Delta t \\
\widetilde{\varepsilon}(t_i) &=& X(t_i) - \mathbb E^{\mathbb Q} [X(t_i) | \mathcal F(t_{i-1})].
\end{eqnarray*}
Then $$X(t_i) = r \Delta t - \widetilde{\gamma}(t_i) + \widetilde{\varepsilon}(t_i).$$

\begin{theorem}\label{Thm:martingale}
The discounted stock price process is a $\mathbb Q$-martingale, i.e., for $0\leq u \leq t \leq T$
$$\mathbb E^{\mathbb Q}[e^{-rt}S(t)|\mathcal F(u)] = e^{-ru} S(u).$$
\end{theorem}
\begin{proof}
By the tower property it suffices to consider the case that $t_{i-1}\leq u <t\leq t_i$.
Since $\{N_{\pm}(t) - N_{\pm}(u)\}| \mathcal F(u) $ have intensities $\widetilde \lambda_{\pm}(t_{i-1})(t-u)$ under $\mathbb Q$, we have
\begin{eqnarray*}
& &\mathbb E^{\mathbb Q}[S(t) | \mathcal F(u)]\\
&=& \mathbb E^{\mathbb Q} [S(u)\exp \{ \delta( N_{+}(t) - N_{+}(u) - N_{-}(t) + N_{-}(u)) \} | \mathcal F(u) ] \\
&=& S(u) \mathbb E^{\mathbb Q} [ \exp \{ \delta( N_{+}(t) - N_{+}(u)) \} | \mathcal F(u) ] \,
\mathbb E^{\mathbb Q} [ \exp \{ -\delta( N_{-}(t) - N_{-}(u)) \} | \mathcal F(u) ] \\
&=& S(u) \exp \{ \widetilde \lambda_{+}(t_{i-1}) (t-u) ( e^\delta -1) \} \exp \{ \widetilde \lambda_{-}(t_{i-1}) (t-u) ( e^{-\delta} -1) \} \\
&=& S(u) e^{r(t-u)}.
\end{eqnarray*}
\end{proof}

Basically the GARCH intensity option pricing model has flexibility in changing the volatility according to the measure change.
This allows us to construct a GARCH intensity option pricing model consistent with volatility spread \cite{bakshi2006theory}.

If we consider a special case when
 $\widetilde\lambda_{+} $ and $\widetilde\lambda_{-} $ satisfy an additional condition
\begin{eqnarray}\label{Eq:variance-preserving}
\widetilde\lambda_{+}(t) + \widetilde\lambda_{-}(t) = \lambda_{+}(t) + \lambda_{-}(t).
\end{eqnarray}
Since  $D(t)=0$, we have
\begin{eqnarray*}
Z(T)&=& \exp  \bigg( \sum_{i=1}^{N} \bigg\{ \left(N_+(t_{i}) - N_+(t_{i-1})\right)\log \frac{\widetilde\lambda_+(t_{i-1})}{\lambda_+(t_{i-1})}\\
&+& \left(N_-(t_{i}) - N_-(t_{i-1})\right)\log \frac{\widetilde\lambda_-(t_{i-1})}{\lambda_-(t_{i-1})}\bigg\} \bigg) .
\end{eqnarray*}
Since $\delta^{2}(\widetilde\lambda_{+}(t_{i-1}) + \widetilde\lambda_{-}(t_{i-1})) \Delta t$ is the conditional variance of $X(t_i)$ under $\mathbb Q$
given  $\mathcal F(t_{i-1})$, we have
$$\Var^{\mathbb Q}(X(t_i) | \mathcal F(t_{i-1}))= \Var^{\mathbb P}(X(t_i) | \mathcal F(t_{i-1}))$$
for  $1\leq i \leq N$.
We plot implied volatility smile in Figure~\ref{Fig:Implied_Vol}.
Using the parameter setting in Table~\ref{Table:estimation2} in Monte Carlo method,
we generate $5\times 10^6$ sample price paths.
The implied volatility is obtained using Black-Scholes-Merton model.
The left panel is for the implied volatility with 30 days maturity
and the right panel is for 60 days maturity.

\begin{figure}
\begin{center}
\includegraphics[width=5cm]{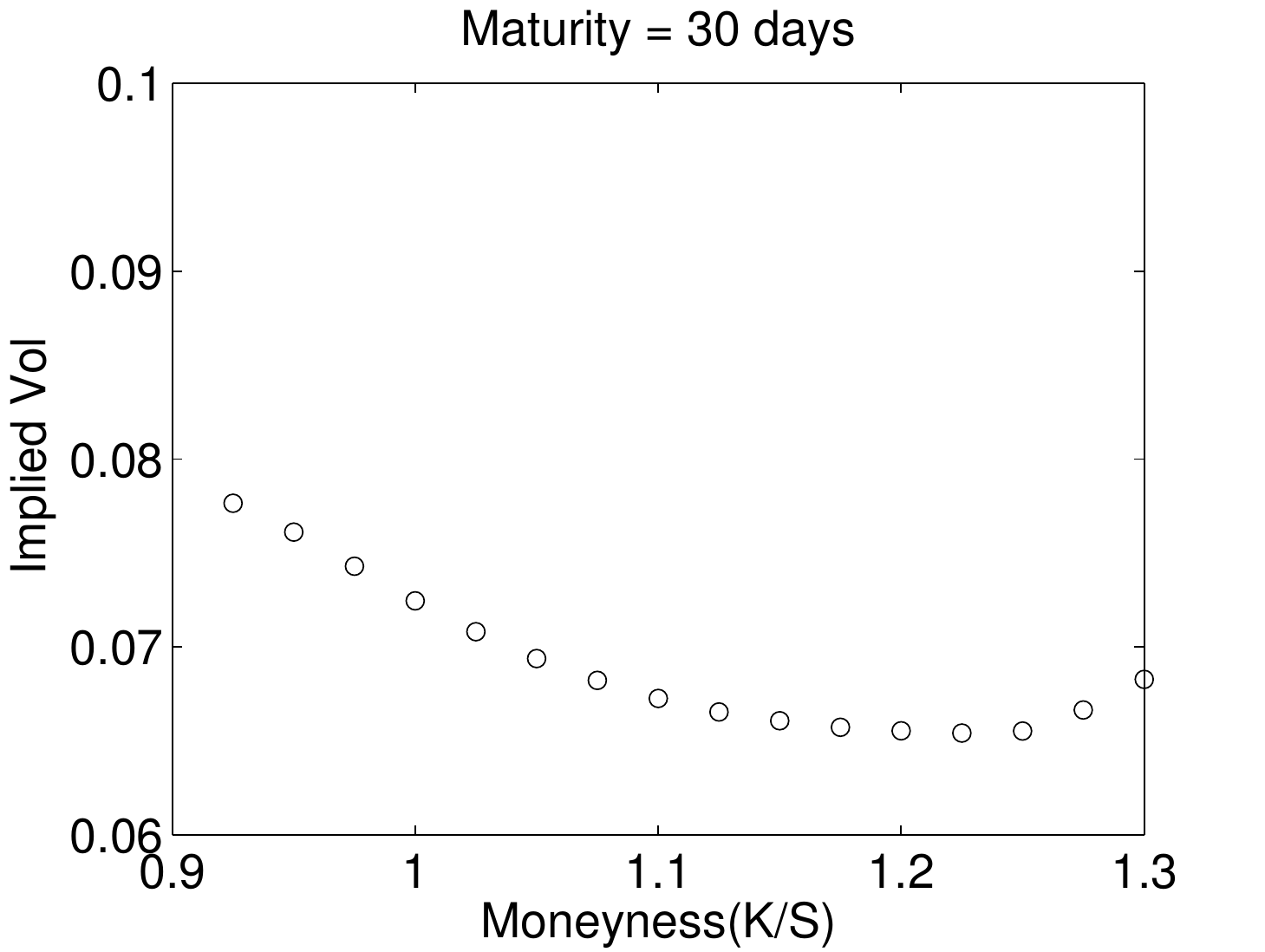}
\quad
\includegraphics[width=5cm]{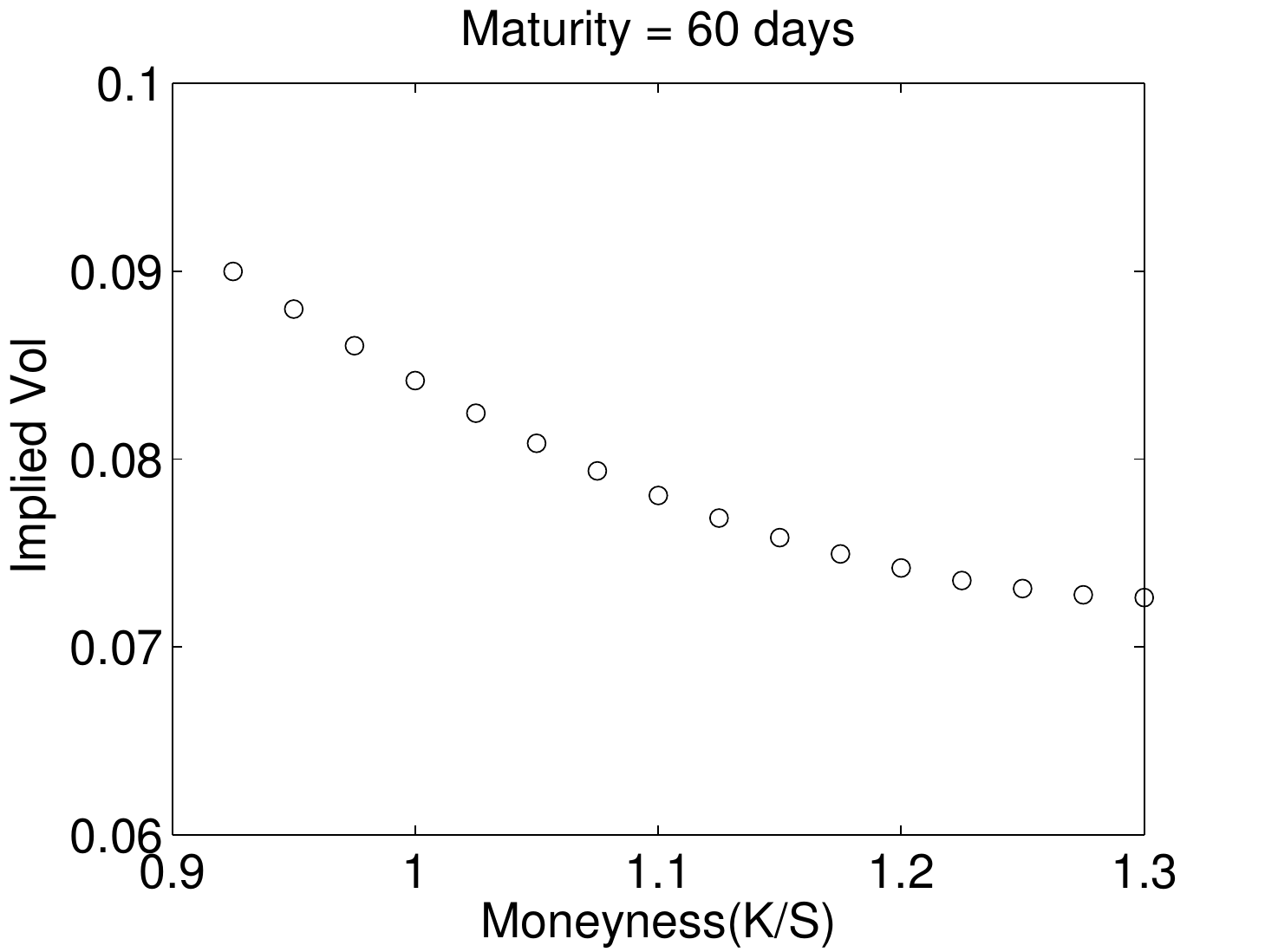}
\end{center}
\caption{GJR: Volatility smile: maturity 30 days and 60 days (from left to right)}
\label{Fig:Implied_Vol}
\end{figure}

\begin{table}
\caption{Parameter setting for GJR intensity model}\label{Table:estimation2}
$$
\begin{array}{cc}
\hline
\delta & 2.0\times 10^{-3} \\
\hline
\omega_{+} & 8.50\times 10^{-2} \\
\beta_{+}  & 9.39\times 10^{-1}  \\
\alpha_{+} & 9.79\times 10^{2} \\
\gamma_{+} & 1.09\times 10^{4}  \\	
\omega_{-} & 7.28\times 10^{-2} \\
\beta_{-}  & 9.42\times 10^{-1} \\
\alpha_{-} & 8.49\times 10^{2} \\
\gamma_{-} & 1.07\times 10^{4} \\	
\hline
\end{array}
$$
\end{table}

\begin{remark}
Assume that $\lambda_{\pm}(t) = \lambda_{\pm}$ and $\widetilde \lambda_{\pm}(t) = \widetilde \lambda_{\pm}$
for some constants $\lambda_{\pm}$ and $\widetilde \lambda_{\pm}$ for every $t$,
and assume that the European call option price at time $t$ is given by
$$ c(t, S(t)) = \mathbb E^{\mathbb Q} [ e^{-r(T-t)} (S(T) - K)^+ | \mathcal F(t)]$$
where $\mathbb Q$ is the conditional variance preserving measure.
By I\^{t}o's formula, we have
\begin{eqnarray*}
e^{-r t} c(t,S(t))
&=&  c(0,S(0)) + A(t)  \\
&+& \int_{0}^{t}  \left[ c(u, e^\delta S(u)) -  c(u, S(u))  \right] \D (N_{+}(u) - \widetilde \lambda_{+}u)\\
&+& \int_{0}^{t}  \left[ c(u, e^{-\delta} S(u)) -  c(u, S(u)) \right] \D (N_{-}(u) - \widetilde \lambda_{-}u)
\end{eqnarray*}
where
\begin{eqnarray*}
A(t) &=& \int_{0}^{t} e^{-ru} \big( -r c(u, S(u)) + \frac{\partial c}{\partial t}(u, S(u))
+ \widetilde \lambda_{+} \left[ c(u, e^\delta S(u)) -  c(u, S(u)) \right] \\
&+& \widetilde \lambda_{-}  \left[ c(u, e^{-\delta} S(u)) -  c(u, S(u)) \right] \big) \D u .
\end{eqnarray*}
Note that $e^{-rt}  c(t, S(t))$,
$N_{+}(t) - \widetilde \lambda_{+}t$  and $N_{-}(t) - \widetilde \lambda_{-}t$ are $\mathbb Q$-martingales.
Hence
$A(t)$ is a martingale, and the integrand should be zero.
Thus $c$ satisfies
$$
 - r c(t, S) + \frac{\partial c}{\partial t} (t,S) + \widetilde\lambda_+\left[c(t, e^\delta S) - c(t, S)\right]
+ \widetilde\lambda_-\left[c(t, e^{-\delta} S)-c(t, S) \right] = 0.
$$
For small $\delta$, consider the following approximations:
$$
c(t, e^{\pm\delta} S) - c(t, S) \approx  (e^{\pm\delta}- 1)S \frac{\partial c}{\partial S}(t, S) +
\frac{1}{2}(e^{\pm\delta} - 1)^2 S^2 \frac{\partial^2 c}{\partial S^2}(t, S).
$$
Then
\begin{eqnarray*}
- r c(t, S) + \frac{\partial c}{\partial t} (t,S) &+& \widetilde\lambda_+\left[(e^\delta - 1)S \frac{\partial c}{\partial S}(t, S) + \frac{1}{2}(e^{\delta} - 1)^2 S^2 \frac{\partial^2 c}{\partial S^2}(t, S)\right] \\
&+& \widetilde\lambda_-\left[(e^{-\delta} - 1)S \frac{\partial c}{\partial S}(t, S) + \frac{1}{2}(e^{-\delta} - 1)^2 S^2 \frac{\partial^2 c}{\partial S^2}(t, S) \right] \\
\approx - r c(t, S) + \frac{\partial c}{\partial t} (t,S) &+& \widetilde\lambda_+\left[(e^\delta - 1)S \frac{\partial c}{\partial S}(t, S) + \frac{1}{2}\delta^2 S^2 \frac{\partial^2 c}{\partial S^2}(t, S)\right] \\
&+& \widetilde\lambda_-\left[(e^{-\delta} - 1)S \frac{\partial c}{\partial S}(t, S) + \frac{1}{2}\delta^2 S^2 \frac{\partial^2 c}{\partial S^2}(t, S) \right] \approx 0
\end{eqnarray*}
and by Eqs.~\eqref{Eq:risk-free-rate} and \eqref{Eq:variance-preserving},
we obtain Black-Scholes-Merton PDE
$$- r c(t, S) + \frac{\partial c}{\partial t} (t,S) + r S \frac{\partial c}{\partial S}(t, S) + \frac{1}{2} h S^2 \frac{\partial^2 c}{\partial S^2}(t, S) = 0$$
where $ h = \delta^2(\lambda_{+} + \lambda_{-})$.
\end{remark}

\section{A generalization of GARCH intensity model}\label{Sect:Generalization}
The goal of this section is to provide the generalized version for the previous model in which
it is assumed to be that the sizes of stock price changes affected by news are represented by independent and identically distributed random variables.

\begin{assumption}\label{Assumption2}
The assumptions (i)--(v) for asset price process are the same as in Assumption~\ref{Assumption}(i)--(v) and we need an additional condition:\\
(vi) (Asset price) The asset price $S(t)$ satisfies
$$ S(t) = S(0) \exp\left(\sum_{j=1}^{N_+(t)}\delta_{+,j} - \sum_{j=1}^{N_-(t)}\delta_{-,j}\right) $$
for some i.i.d. random variables $\delta_{+,j}>0$ and $\delta_{-,j}>0$, $j\geq 1$.
\end{assumption}

\begin{lemma}\label{Lemma5}
Under Assumption~\ref{Assumption2}, we have
$$\mathbb E [X(t_i) | \mathcal F(t_{i-1})]  =  (\bar \delta_+ \lambda_{+}(t_{i-1}) - \bar \delta_- \lambda_{-}(t_{i-1})) \Delta t$$
and
$$ \Var(X(t_i) |\mathcal F(t_{i-1}) )= (\overline{\delta_+^2} \lambda_+(t_{i-1}) + \overline{\delta_-^2} \lambda_-(t_{i-1}))\Delta t$$
where
$$\bar\delta_\pm = \E [\delta_\pm]$$
and
$$\overline{\delta_\pm^2} = \E [\delta_\pm^2].$$
\end{lemma}
\begin{proof}
Recall that
$$X(t_i) =  \sum_{j=N_+(t_{i-1})+1}^{N_+(t_i)}\delta_{+,j} - \sum_{j=N_-(t_{i-1})+1}^{N_-(t_i)}\delta_{-,j} .$$
Now use the fact that each conditional distribution
$$\sum_{j=N_\pm(t_{i-1})+1)}^{N_\pm(t_i)}\delta_{\pm,j}\big| \F(t_{i-1}) $$
is a compound Poisson distribution.
\end{proof}

Note that depending on distributions of $\delta_+$ and $\delta_-$, the model allows the skewness in the conditional distribution of log-return $X(t)$.
For example, if the tail of the distribution of $\delta_-$ is fatter than the tail of the distribution of $\delta_+$,
then the conditional distribution of log-return $X(t)$ is negatively skewed.

\begin{definition}
We define
\begin{eqnarray*}
\mu(t_i) &=& (\phi_{\delta_+} - 1) \lambda_+ (t_{i-1})\Delta t + (\phi_{\delta_-} - 1) \lambda_- (t_{i-1})\Delta t \\
\gamma(t_i) &=& (\phi_{\delta_+} - 1 - \bar \delta_+) \lambda_+ (t_{i-1})\Delta t + (\phi_{\delta_-} - 1 - \bar \delta_-) \lambda_- (t_{i-1})\Delta t \\
\varepsilon(t_i) &=& X(t_i) - \mathbb E [X(t_i) | \mathcal F(t_{i-1})]
\end{eqnarray*}
where
$$\phi_{\delta_{+}} = \E [e^{\delta_{+}}], \textrm{ and } \phi_{\delta_{-}} = \E [e^{-\delta_{-}}]$$
respectively.
\end{definition}
Note that, by direct computation, we have
\begin{equation}\label{Eq:components}
X(t_i) = \mu(t_i) - \gamma(t_i) + \varepsilon(t_i).
\end{equation}

The mean correction factor $\gamma(t_i)$ appears because we model with log-return and this implies that
\begin{eqnarray*}
& & \E[\exp (\varepsilon(t_i)) | \F(t_{i-1})]\\
&=& \E [ \exp \left\{ -\E[X(t_i)|\mathcal F(t_{i-1}) ] + X(t_i) \right\} | \F(t_{i-1})] \\
&=& \exp\{ -\E[X(t_i)|\mathcal F(t_{i-1})]\}\\
& &\times \quad \E\left[ \exp \left( \sum_{j=N_+(t_{i-1})+1}^{N_+(t_{i})} \delta_{+,j} \right) \bigg| \F(t_{i-1}) \right] \\
& &\times \quad \E \left[ \exp \left( \sum_{j=N_-(t_{i-1})+1}^{N_-(t_{i})} \delta_{-,j} \right) \bigg| \F(t_{i-1}) \right]
\end{eqnarray*}
and using the property of compound Poisson distributions, we have
\begin{eqnarray*}
& & \E[\exp (\varepsilon(t_i)) | \F(t_{i-1})] \\
&=& \exp\{(\phi_{\delta_+} - 1 - \bar \delta_+) \lambda_+ (t_{i-1})\Delta t + (\phi_{\delta_-} - 1 - \bar \delta_-) \lambda_- (t_{i-1})\Delta t\}\\
&=& \exp(\gamma(t_i)).
\end{eqnarray*}
The one step ahead expectation of future stock price can be represented as exponential of drift term $\mu(t_i)$ multiplied by current stock price,
that is
\begin{eqnarray*}
& &\E[S(t_i) | \F(t_{i-1})] \\
&=&S(t_{i-1})\exp(\mu(t_i) - \gamma(t_i))\E[\exp(\varepsilon(t_i)) | \F(t_{i-1})]\\
&=&S(t_{i-1})\exp(\mu(t_i)).
\end{eqnarray*}

The derivation of equivalent martingale measure for extended version of GRACH intensity model is similar to the previous version.

\begin{definition}\label{Def:RadonNikodym2}
Let $f_\pm$ be probability density functions of $\delta_\pm$, respectively
and $\widetilde f_\pm$ are some probability density functions
(which are desired probability density functions of $\delta_\pm$ under equivalent martingale measure). \\
(i) Let
\begin{eqnarray*}
& &\widetilde\phi_{\delta_{+}} = \E^{\Q}[e^{\delta_{+}}] = \int_{-\infty}^{\infty} e^x \widetilde f_{+}(x) \D x,\\
& &\widetilde\phi_{\delta_{-}} = \E^{\Q}[e^{-\delta_{-}}] = \int_{-\infty}^{\infty} e^{-x} \widetilde f_{-}(x) \D x
\end{eqnarray*}
and
$\widetilde\lambda_{+} $ and $\widetilde\lambda_{-} $ be two r.c.l.l. adapted step processes satisfying the equation
\begin{equation}\label{Eq:riskEquation}
(\widetilde\phi_{\delta_+} -1)\widetilde\lambda_{+}(t_{i}) + (\widetilde\phi_{\delta_-} -1)\widetilde\lambda_{-}(t_{i}) = r\Delta t
\end{equation}
for each $i$, and $$\widetilde\lambda_{\pm}(t) = \widetilde\lambda_{\pm}(t_{i-1})$$ for $t_{i-1} \leq t < t_{i}$.
\\
(ii) Suppose that $\widetilde\lambda_{+}(t) $ and $\widetilde\lambda_{-}(t) $ are positive processes.
Let
$$\kappa_{i-1} = \widetilde\lambda_+(t_{i-1}) - \lambda_+(t_{i-1}) + \widetilde \lambda_-(t_{i-1}) - \lambda_-(t_{i-1}), $$
$$ Q_{i-1}(t) = \sum_{j=N_+(t_{i-1})+1}^{N_+(t)} \log \frac{\lambda_{+}(t_{i-1}) f(\delta_{+,j})}{\widetilde\lambda_{+}(t_{i-1}) \widetilde f(\delta_{+,j})}
+ \sum_{j=N_-(t_{i-1})+1}^{N_-(t)} \log \frac{\lambda_{-}(t_{i-1}) f(\delta_{-,j})}{\widetilde\lambda_{-}(t_{i-1}) \widetilde f(\delta_{-,j})},
$$
for integer $0 \leq i \leq N$.
Define
$$Z(0) =1$$
and
$$Z(t) = Z(t_{i-1}) \exp ( - \kappa_{i-1} (t-t_{i-1}) -  Q_{i-1}(t))$$
for $t_{i-1} < t \leq t_{i}$, recursively.
\end{definition}

\begin{remark}
Note that $\mathcal F(t_{i-1})$-measurable random variable $\kappa_{i-1}$ is zero
when the conditional variance of return distribution of the martingale measure is equal to the variance under physical measure.
$Q_i(t)$ depends on $\mathcal F(t)$ random variables $N_{\pm}$ and $\delta$.
\end{remark}

\begin{lemma}\label{Lemma:RadonNikodym2}
For $0 \leq t \leq T$,  $$\mathbb E[Z(t)]=1.$$
\end{lemma}
\begin{proof}
For $t_{i-1} < t \leq t_i$, define
$$
Z_{+,i}(t) \\
= \exp\left\{ (\lambda_{+}(t_{i-1}) - \widetilde\lambda_{+}(t_{i-1}))(t-t_{i-1}) +
\sum_{j=N_+(t_{i-1})+1}^{N_+(t)} \log \frac{\widetilde\lambda_{+}(t_{i-1}) \widetilde f(\delta_{+,j})}{\lambda_{+}(t_{i-1}) f(\delta_{+,j})} \right\}
$$
$$
Z_{-,i}(t) \\
= \exp\left\{ (\lambda_{-}(t_{i-1}) - \widetilde\lambda_{-}(t_{i-1}))(t-t_{i-1}) +
\sum_{j=N_-(t_{i-1})+1}^{N_-(t)} \log \frac{\widetilde\lambda_{-}(t_{i-1}) \widetilde f(\delta_{-,j})}{\lambda_{-}(t_{i-1}) f(\delta_{-,j})} \right\}.
$$
Then $Z_{+,i}(t)$ and $Z_{-,i}(t)$ are $\mathcal F(t)$-measurable and
satisfy
$$\mathbb E[Z_{+,i}(t)| \F(t_{i-1})]=1$$
and
$$\mathbb E[Z_{-,i}(t)|F(t_{i-1})]=1.$$
Note that
\begin{eqnarray*}
& & Z_{+,i}(t)Z_{-,i}(t) = \frac{Z(t)}{Z(t_{i-1})}.
\end{eqnarray*}
Since $N_{+}(t) - N_{+}(t_{i-1})$ and $N_{-}(t) - N_{-}(t_{i-1})$ are conditionally independent upon $\mathcal F(t_{i-1})$, we have
$$
\mathbb E\left[\frac{Z(t)}{Z(t_{i-1})} \bigg| \mathcal F(t_{i-1})\right]
=\mathbb E[Z_{+,i}(t) | \mathcal F(t_{i-1}] \; \mathbb E[Z_{-,i}(t) | \mathcal F_{i-1}] =1.
$$
Finally, for $0 < t \leq T$, we have
\begin{eqnarray*}
\mathbb E[Z(t)] &=& \mathbb E\left[\frac{Z(t)}{Z(t_{i-1})}\frac{Z(t_{i-1})}{Z(t_{i-2})} \dots \frac{Z(t_2)}{Z(t_1)} \frac{Z(t_1)}{Z(0)}\right] \\
&=& \mathbb E \left[\mathbb E\left[\frac{Z(t)}{Z(t_{i-1})} \frac{Z(t_{i-1})}{Z(t_{i-2})} \dots \frac{Z(t_2)}{Z(t_1)} \frac{Z(t_1)}{Z(0)} \bigg| \mathcal F(t_{i-1})\right]\right] \\
&=& \mathbb E\left[\frac{Z(t_{i-1})}{Z(t_{t-2})} \dots \frac{Z(t_2)}{Z(t_1)} \frac{Z(t_1)}{Z(0)} \right] \\
&=& \quad \vdots \\
&=& \mathbb E \left[ \frac{Z(t_1)}{Z(0)} \right] = 1.
\end{eqnarray*}
\end{proof}

Since $\mathbb E[Z(T)] =1$, we use $Z(T)$ in Definition~\ref{Def:RadonNikodym2} to construct a new probability measure $\mathbb Q$.
We define
\begin{equation}\label{Eq:measure2}
\mathbb Q(A) = \int_{A} Z(T) \D \mathbb P \quad \textrm{for }A \in \mathcal F.
\end{equation}

\begin{lemma}\label{Lem:newIntensity2}
Under the measure $\mathbb Q$ in \eqref{Eq:measure2}, for every $t_{i-1}<t\leq t_i$,
the conditional distributions
$$\left(N_{+}(t)-N_+(t_{i-1})\right)|\F(t_{i-1})$$
and $$\left(N_{-}(t)-N_{-}(t_{i-1}) \right)|\F(t_{i-1})$$ are Poisson distributions
with  new intensities $\widetilde \lambda_{+}(t_{i-1})$ and $\widetilde \lambda_{-}(t_{i-1})$, respectively.
Moreover, $\delta_{+,j}$ have probability density function $\widetilde f$ under $\Q$.
\end{lemma}

\begin{proof}
Define $Z_{\pm,i}(t)$ as in the proof of Lemma~\ref{Lemma:RadonNikodym2}.
For a constant $u$, we have
\begin{eqnarray*}
& & \E^{\Q}[ \exp\{ u(N_+(t)-N_+(t_{i-1})) \} | \F(t_{i-1})] \\
&=& \E^{\mathbb P} \left[ \exp \left\{ u(N_+(t)-N_+(t_{i-1})) \right\} \frac{Z(t)}{Z(t_{i-1})} \bigg| \F(t_{i-1})\right]\\
&=& \E^{\mathbb P}[ \exp\{u(N_+(t)-N_+(t_{i-1}))\} Z_{+,i}(t)Z_{-,i}(t) | \F(t_{i-1})]\\
&=& \E^{\mathbb P}[ \exp\{u(N_+(t)-N_+(t_{i-1}))\} Z_{+,i}(t) | \F(t_{i-1})] \;\E^{\mathbb P}[ Z_{-,i}(t)| \F(t_{i-1}) ]\\
&=& \E^{\mathbb P}[ \exp\left\{u(N_+(t)-N_+(t_{i-1}))\right\} Z_{+,i}(t) | \F(t_{i-1})]
\end{eqnarray*}
and
\begin{eqnarray*}
& &\E^{\mathbb P} [ \exp\left\{u(N_+(t)-N_+(t_{i-1}))\right\} Z_{+,i}(t) | \F(t_{i-1})]\\
&=& \E^{\mathbb P} \bigg[ \exp\{u(N_+(t)-N_+(t_{i-1})) + (\lambda_+(t_{i-1}) - \widetilde\lambda_+ (t_{i-1}))(t-t_{i-1})\}  \\
& &\times\quad \exp\left\{  \sum_{j=N_+(t_{i-1})+1}^{N_+(t)} \log \frac{\widetilde\lambda_{+}(t_{i-1}) \widetilde f(\delta_{+,j})}{\lambda_{+}(t_{i-1}) f(\delta_{+,j})} \right\} \bigg| \F(t_{i-1})\bigg] \\
&=& \exp \{ (\lambda_+ (t_{i-1}) - \widetilde\lambda_+ (t_{i-1}))(t-t_{i-1}) \} \\
& &\times\quad \E^{\mathbb P} \left[ \exp\left\{\sum_{j=N_+(t_{i-1})+1}^{N_+(t)} \left( u+  \log \frac{\widetilde\lambda_{+}(t_{i-1}) \widetilde f(\delta_{+,j})}{\lambda_{+}(t_{i-1}) f(\delta_{+,j})}\right) \right\} \bigg| \F(t_{i-1})\right].
\end{eqnarray*}
Put
$$ Y_{+,j} = u + \log \frac{\widetilde\lambda_{+}(t_{i-1}) \widetilde f(\delta_{+,j})}{\lambda_{+}(t_{i-1}) f(\delta_{+,j})} $$
and let $\Phi$ be the conditional moment generating function of $Y_{+,j}$, i.e.
$$\Phi(z) = \mathbb E^{\mathbb P} [ \exp( z Y_{+,j} ) | \F(t_{i-1}) ].$$
Note that
\begin{eqnarray*}
\Phi(1) &=& \frac{\widetilde\lambda_{+}(t_{i-1})} {\lambda_{+}(t_{i-1}) } e^{u}\E^{\mathbb P} \left[ \frac{\widetilde f(\delta_{+,j})}{f(\delta_{+,j})} \bigg| \F(t_{i-1}) \right] \\
&=& \frac{\widetilde\lambda_{+}(t_{i-1})} {\lambda_{+}(t_{i-1}) } e^{u} \int_{-\infty}^{\infty} \frac{\widetilde f(x)}{f(x)} f(x) \D x \\
&=& \frac{\widetilde\lambda_{+}(t_{i-1})} {\lambda_{+}(t_{i-1}) } e^{u}.
\end{eqnarray*}
Then
$$
\E^{\mathbb P} \left[ \exp\left\{\sum_{j=N_+(t_{i-1})+1}^{N_+(t)}  Y_{+,j} \right\} \bigg| \F(t_{i-1})\right]
= \exp \left\{ \lambda_+(t_{i-1}) \left( \Phi(1) - 1\right)(t-t_{i-1}) \right\}.
$$
Note that
\begin{eqnarray*}
& &\E^{\mathbb P} [ \exp\left\{u(N_+(t)-N_+(t_{i-1}))\right\} Z_{+,i}(t) | \F(t_{i-1})]\\
&=& \exp \left\{ \left(\lambda_+ (t_{i-1}) - \widetilde\lambda_+ (t_{i-1}) + \lambda_+(t_{i-1}) \left( \frac{\widetilde\lambda_{+}(t_{i-1})} {\lambda_{+}(t_{i-1}) } e^{u} - 1\right) \right)(t-t_{i-1}) \right\} \\
&=& \exp \{ \widetilde\lambda_+(t_{i-1})(e^u - 1)(t-t_{i-1}) \},
\end{eqnarray*}
which is the moment generating function of a Poissson distribution with intensity $\widetilde \lambda_+(t_{i-1})$.
Hence $N_{+}(t)-N_+(t_{i-1})|\F(t_{i-1})$ is a Poissson distribution with intensity $\widetilde \lambda_+(t_{i-1})$.
For $N_-(t) - N_-(t_{i-1})$, the proof is the same.

For the remaining part of lemma, to figure out the distribution of $\delta_{\pm, j}$ under $\Q$,
it is enough to check the distribution of $\delta_{+, 1}$ under $\Q$.
Without loss of generality, assume that first jump, i.e. the occurrence time of random variable $\delta_{+, 1}$ is less than $t_1$.
Thus, for a constant $u$, we have
\begin{eqnarray*}
& & \E^{\Q} [e^{u\delta_{+,1}}]\\
&=& \E^{\mathbb P}[e^{u\delta_{+,1}} Z(T) ]\\
&=&  \E^{\mathbb P}[e^{u\delta_{+,1}} Z_{+,1}(t_1)Z_{-,1}(t_1)Z(T)/Z(t_1) ] \\
&=& \E^{\mathbb P} \bigg[ \exp \left\{ (\lambda_+(t_0) - \widetilde \lambda_-(t_0))\Delta t + N_+(t_1)\log \frac{\lambda_+(t)}{\lambda_-(t)}  \right\} \\
& &\times\quad e^{u\delta_+,1}\frac{\widetilde f(\delta_{+,1})}{f(\delta_{+,1})}  \prod_{j=2}^{N_+(t_1)} \frac{\widetilde f(\delta_{+,j})}{f(\delta_{+,j})} \bigg].
\end{eqnarray*}
The last equality holds since given random variables independent and since
$$\E[Z_{-,1}(t_1)] = \E \left[  \frac{Z(T)}{Z(t_1)} \right] =1.$$
Furthermore, because of independency and the fact that
$$\E^{\mathbb P} \left[ \exp \left\{ (\lambda_+(t_0) - \widetilde \lambda_-(t_0))\Delta t + N_+(t_1)\log \frac{\lambda_+(t)}{\lambda_-(t)}  \right\} \right]=1 $$
(by Lemma~\ref{Lemma:RadonNikodym2} when $\delta$ is constant) and
$$\E^{\mathbb P} \left[ \frac{\widetilde f(\delta_{+,j})}{f(\delta_{+,j})} \right] =1 $$ for each $j$.
Finally we have
\begin{eqnarray*}
\E^{\Q} [e^{u\delta_{+,1}}] &=& \E^{\mathbb P} \left[e^{u\delta_+,1}\frac{\widetilde f(\delta_{+,1})}{f(\delta_{+,1})} \right] \\
&=& \int_{-\infty}^{\infty} e^{ux} \frac{\widetilde f(x)}{f(x)} f(x) \D x \\
&=& \int_{-\infty}^{\infty} e^{ux} \widetilde f(x) \D x,
\end{eqnarray*}
which implies $\delta_{+,1}$ has a probability density function $\widetilde f$ under $\Q$.
For general $j$, the proof is the same.
\end{proof}

The choice for risk-neutral distributions $\widetilde f_+$ and $\widetilde f_-$ are related to the skewness in conditional distribution of log-return under risk-neutral measure.
This is similar to the fact that distributions $f_+$ and $f_-$ are related to the skewness under physical measure.
Now we show that under the equivalent martingale measure $\Q$, the discounted stock price process is a martingale.
\begin{theorem}\label{Thm:martingale2}
Under the measure $\mathbb Q$ defined by \eqref{Eq:measure2}, we have
$$\mathbb E^{\mathbb Q}[e^{-r(t-u)}S(t)|\F(u)] = S(u)$$
for $0<u<t<T$.
\end{theorem}

\begin{proof}
Take $u,t$ such that $t_{i-1}\leq u<t\leq t_i$.
\begin{eqnarray*}
& &\mathbb E^{\mathbb Q}[S(u) | \mathcal F(t)] \\
&=& \mathbb E^{\mathbb Q} \left[S(t)\exp \left( \sum_{j=N_+(u)+1}^{N_+(t)} \delta_{+, j} - \sum_{j=N_-(u)+1}^{N_-(t)} \delta_{-, j} \right) \bigg| \mathcal F(t) \right] \\
&=&  \mathbb E^{\mathbb Q} \left[ S(t) \exp \left( \sum_{j=N_+(u)+1}^{N_+(t)} \delta_{+, j} - \sum_{j=N_-(u)+1}^{N_-(t)} \delta_{-, j} \right) \frac{Z(t)}{Z(u)} \bigg| \mathcal F(t) \right] \\
&=& S(t) \exp \{ (\lambda_+(t_{i-1}) - \widetilde \lambda_+(t_{i-1}) + \lambda_-(t_{i-1}) - \lambda_-(t_{i-1})) (t-u) \} \\
& &\times\quad \E^{\mathbb P} \left[ \exp \left\{ \sum_{j=N_+(u)+1}^{N_+(t)} \left( \delta_{+, j} + \log \frac{\widetilde\lambda_+(t_{i-1})\widetilde f(\delta_{+,j})}{\lambda_+(t_{i-1}) f(\delta_{+,j})} \right) \right\} \bigg| \mathcal F(t) \right]\\
& &\times\quad \E^{\mathbb P} \left[ \exp \left\{ - \sum_{j=N_-(u)+1}^{N_-(t)} \left( \delta_{-, j} + \log \frac{\widetilde\lambda_-(t_{i-1})\widetilde f(\delta_{-,j})}{\lambda_-(t_{i-1}) f(\delta_{-,j})} \right) \right\} \bigg| \mathcal F(t) \right].
\end{eqnarray*}
Let $\phi_{W_+}$ be the conditional moment generating function of
$$ W_{+,j} = \delta_{+, j} + \log \frac{\widetilde\lambda_+(t_{i-1})\widetilde f(\delta_{+,j})}{\lambda_+(t_{i-1}) f(\delta_{+,j})}$$
given filtration $\F(t)$.
Then
\begin{eqnarray*}
\phi_{W_+}(1) &=& \E^{\mathbb P} [e^{W_{+,j}} | \F(t)]
= \frac{\widetilde\lambda_{+}(t_{i-1})} {\lambda_{+}(t_{i-1}) } \, \E^{\mathbb P} \left[ e^\delta_{+, j} \frac{\widetilde f(\delta_{+,j})}{f(\delta_{+,j})} \bigg| \F(t) \right] \\
&=& \frac{\widetilde\lambda_{+}(t_{i-1})} {\lambda_{+}(t_{i-1}) }  \int_{-\infty}^{\infty} e^x \frac{\widetilde f(x)}{f(x)} f(x) \D x \\
&=& \frac{\widetilde\lambda_{+}(t_{i-1})} {\lambda_{+}(t_{i-1}) } \, \widetilde\phi_{\delta_+}
\end{eqnarray*}
where $\widetilde\phi_{\delta_+}$ is defined in Definition~\ref{Def:RadonNikodym2}.
Hence
\begin{eqnarray*}
& & \mathbb E^{\mathbb P} \left[ \exp \left\{ \sum_{j=N_+(u)+1}^{N_+(t)} \left( \delta_{+, j} + \log \frac{\widetilde\lambda_+(t_{i-1})\widetilde f(\delta_{+,j})}{\lambda_+(t_{i-1}) f(\delta_{+,j})} \right) \right\} \bigg| \F(u)\right] \\
&=& \exp \{ \lambda_+(t_{i-1})(t-u)(\phi_{W_+}(1)-1)\} \\
&=& \exp \{ (\widetilde\lambda_+(t_{i-1})\widetilde\phi_{\delta_+} - \lambda(t_{i-1}))(t-u)\}
\end{eqnarray*}
and
\begin{eqnarray*}
& &\mathbb E^{\mathbb Q}[S(u) | \mathcal F(t)] \\
&=& S(t) \exp \left\{ (\lambda_+(t_{i-1}) - \widetilde \lambda_+(t_{i-1}) + \lambda_+(t_{i-1}) - \widetilde \lambda_+(t_{i-1}))(t-u) \right\} \\
& &\times \exp \{ (\widetilde \lambda_{+}(t_{i-1}) \widetilde\phi_{\delta_+} - \lambda_{+}(t_{i-1}) + \widetilde \lambda_{-}(t_{i-1}) \widetilde\phi_{\delta_+} - \lambda_{-}(t_{i-1}))(t-u) \} \\
&=& S(t) \exp \{ (\widetilde \lambda_{+}(t) \widetilde\phi_{\delta_+} - \widetilde\lambda_{+}(t) + \widetilde \lambda_{-}(t) \widetilde\phi_{\delta_+} - \widetilde\lambda_{-}(t))(t-u)\} \\
&=& S(t) e^{r(t-u)}.
\end{eqnarray*}
The last equality is due to Definition~\ref{Def:RadonNikodym2}(i).
By applying the tower property, we obtain the desired result for arbitrary $u$ and $t$.
\end{proof}

\section{Concluding remark}\label{Sect:concl}
The risk-neutral option pricing framework for the GARCH intensity model was introduced.
Equivalent martingale measures are provided and hence the the risk-neutral option price is computed under the measure.
The framework is consistent with the empirical characteristics such as volatility smile and spread.
The theory is easily extended to the generalized version of the GARCH intensity model.

\bibliography{Poisson_intensity_bib}
\bibliographystyle{jbact}
\end{document}